\documentclass[conference,10pt,letterpaper]{IEEEtran}
\IEEEoverridecommandlockouts

\usepackage[T1]{fontenc}
\usepackage[utf8]{inputenc}
\usepackage{cite}
\usepackage{amsmath,amssymb,amsthm,mathtools}
\usepackage{textcomp}
\usepackage{xcolor}
\usepackage{booktabs}
\usepackage{array}
\usepackage{enumitem}
\usepackage{microtype}
\usepackage{balance}
\usepackage[hyphens]{url}
\usepackage[hidelinks,breaklinks]{hyperref}

\newtheorem{theorem}{Theorem}
\newtheorem{lemma}[theorem]{Lemma}
\newtheorem{corollary}[theorem]{Corollary}
\newtheorem{proposition}[theorem]{Proposition}
\theoremstyle{definition}
\newtheorem{definition}{Definition}
\theoremstyle{remark}
\newtheorem{remark}{Remark}

\newcommand{\Stream}[1]{S_{#1}}
\newcommand{\hash}[1]{H(#1)}
\newcommand{\INTER}{\mathit{INTER}}
\newcommand{\epoch}{\mathit{ep}}
\newcommand{\swarm}[1]{\mathcal{W}_{#1}}
\newcommand{\supermaj}{\tfrac{2}{3}}
\newcommand{\rID}{\mathit{rID}}
\newcommand{\PoC}{\Pi}

\setlist[itemize]{leftmargin=1.4em,topsep=2pt,itemsep=1pt,parsep=0pt}
\setlist[enumerate]{leftmargin=1.8em,topsep=2pt,itemsep=1pt,parsep=0pt}

\begin{document}

\title{Intercloud: Eventual Consistency for\\
Decentralised Economies via\\
Chilling-Effect Consensus}

\author{\IEEEauthorblockN{Gregory Magarshak}
\IEEEauthorblockA{IENYC\\
Email: gmagarshak@faculty.ienyc.edu}}

\maketitle

\begin{abstract}
We present \emph{Intercloud}, a decentralised economic network in which streams of private data are secured by Watcher swarms that observe only cryptographic hashes, never plaintext. Intercloud requires no global consensus beyond a single shared random seed per epoch. Security is provided by two mechanisms inherited from the Magarshak Machine~\cite{magarshak2026mm}: (i)~\emph{ripple deduplication} via epoch-stamped identifiers, guaranteeing that reactive execution terminates without global coordination by preventing any ripple from rippling through the same node twice per epoch; and (ii)~\emph{chilling-effect consensus}, in which a swarm reaches finality by attesting to the \emph{absence of conflicting evidence}, not by voting between alternatives. Any conflicting signed attestation automatically yields a self-certifying Proof of Corruption, making misbehaviour irrefutably detectable. We prove that execution ripples terminate in bounded time; that a swarm of approximately 35 Watchers assigned by a verifiable random function suffices for high-probability double-spending prevention, matching Hoepman's lower bound~\cite{hoepman2008}; that two correct clients can hold conflicting finality attestations only if the adversary both compromises a supermajority of the assigned swarm and eclipses both clients from all honest nodes; and that Buridan's Principle~\cite{lamport2012buridan} does not apply because the swarm is finite and the consensus question is absence of evidence. Stream states are coloured Green, Yellow, or Red; end users choose their own finality threshold. The single global random oracle required costs $O(N)$ messages \emph{per epoch}, not per transaction. Beyond the consensus layer, we develop a full economic model in which local coins are issued and retired by currency streams with pluggable exchange-rate formulae, and we prove that security weight tracks economic value automatically via a dot-product identity. The coin layer and content layer are strictly separated.
\end{abstract}


\section{Introduction}
\label{sec:intro}

\subsection{The Blockchain Cost Problem}

Consider transporting \$1 by armoured truck. A rational security model assigns one truck to one dollar. A blockchain does the opposite: every validator in the network must verify every transaction, regardless of its value. The global security budget --- billions of dollars of electricity in proof-of-work systems, billions of staked tokens in proof-of-stake systems --- is deployed uniformly for every transaction. A one-dollar transfer and a billion-dollar transfer consume the same validation resources.

Beyond cost, the global ledger model has a privacy problem. Every transaction, every balance, every smart-contract invocation is visible to all participants. Transaction graph analysis can de-anonymise users even when addresses are pseudonymous.

Intercloud addresses both problems simultaneously. Its central architectural choice is that Watcher nodes --- the nodes that provide security --- see only \emph{hashes} of stream states, never plaintext. Security is allocated \emph{proportionally to the square root of economic value}: a stream holding one dollar uses approximately 35 Watchers; a stream holding one million dollars uses $\sqrt{10^6} \approx 1000$ times more Watchers --- roughly 35{,}000. This is the armoured-convoy model made precise: the cost of security scales sub-linearly with the value transported.

\subsection{The Two Core Mechanisms}

\subsubsection{Ripple deduplication}
The Magarshak Machine~\cite{magarshak2026mm} prevents cyclic reactive execution by tagging each ripple with a unique \emph{ripple identifier} ($\rID$). Nodes store $\rID$s in a hash table indexed by $(\rID, \epoch)$, retaining entries for two epoch durations before eviction. Because execution ripples carry fees at each hop, the table remains compact. This mechanism is proven correct for single-node Magarshak Machines; we extend it to distributed Intercloud in Section~\ref{sec:waves}.

\subsubsection{Chilling-effect consensus}
A Watcher swarm for stream $\Stream{i}$ reaches \emph{finality} when a supermajority of its members has each independently verified via gossip that a supermajority of the swarm has attested to the same hash $h$ and seen no conflicting hash during this epoch. Critically, this is not a vote between two candidate values. It is an attestation to the \emph{absence of conflict}. If any two attestations conflict, both signers automatically produce a Proof of Corruption ($\PoC$) --- a self-certifying evidence package requiring no trusted third party to verify. The chilling effect is the deterrence: a Watcher that produces conflicting attestations is immediately and irrevocably convicted.

\subsection{Comparison with Existing Approaches}

\subsubsection{Global-consensus systems}
Nakamoto~\cite{nakamoto2008bitcoin} and proof-of-stake variants require every validator to process every transaction. The validation cost is $O(N)$ per transaction where $N$ is the network size.

\subsubsection{BFT protocols}
PBFT~\cite{castro1999pbft} achieves safety with $f < n/3$ Byzantine nodes but requires $O(n^2)$ messages per consensus round. Tendermint~\cite{buchman2016tendermint} reduces this but still requires global participation per round.

\subsubsection{Hoepman's result}
Hoepman~\cite{hoepman2008} proved that using coin identifiers to assign clerk sets, double-spending can be prevented with clerk sets whose size is independent of the total number of nodes $N$, depending only on security parameters. For standard parameters this yields sets of approximately 35 nodes. Intercloud's Watcher swarms are precisely these clerk sets, with stream identifiers replacing coin identifiers.

\subsubsection{Buridan's Principle}
Lamport~\cite{lamport2012buridan} proved that a binary discrete decision on a continuous-valued input cannot be made in bounded time. We prove in Section~\ref{sec:buridan} that chilling-effect consensus is not subject to this principle because the swarm is finite and the decision is not binary.

\subsection{Paper Organisation}

Section~\ref{sec:background} reviews relevant prior work. Section~\ref{sec:model} defines the Intercloud model. Sections~\ref{sec:waves}--\ref{sec:randao} state and prove the main results. Section~\ref{sec:junior_rewards} formalises the junior-node corruption detection and PoC lottery reward mechanism. Section~\ref{sec:stake} develops the vesting and rational security argument. Sections~\ref{sec:local_coins}--\ref{sec:privacy_separation} develop the local coin economy: emission, retirement, exchange rates, the dot-product security theorem, and the coin--content separation. Section~\ref{sec:zk} presents the zero-knowledge extension. Section~\ref{sec:limitations} discusses limitations and future work. Section~\ref{sec:conclusion} concludes.

\section{Background}
\label{sec:background}

\subsubsection{Magarshak Machine}
A Magarshak Machine~\cite{magarshak2026mm} $\mathcal{MM} = (S, P, A, C, E, R)$ --- the \textsc{spacer} model --- is a formal model for governed, append-only distributed state. Streams $S_i$ are append-only message logs under a publisher's authority; $R$ is a bidirectional relation index --- two tables ($\mathit{relatedTo}$, $\mathit{relatedFrom}$) updated atomically when streams post \texttt{relateTo}/\texttt{unrelateTo} messages; actions $A$ are policy-checked transformations declaring their read and write sets; capabilities $C$ are the sole mechanism for side effects; policy $P$ governs all transitions; and execution engine $E$ is the push-only reactive scheduler driven by four reduction rules (\textsc{Create}, \textsc{Trigger}, \textsc{Execute}, \textsc{Retry}). The MM paper derives 22 structural theorems from these rules, including append-only safety, embarrassing parallelism scaling linearly with publishers, causal consistency, ripple termination via local deduplication, minimal cache invalidation, deterministic replay, and consensus freedom. The Magarshak Machine proves ripple termination for single nodes via the local ripple-log deduplication mechanism. We extend this to distributed networks in Theorem~\ref{thm:termination}.

\subsubsection{Lamport's causal ordering}
Lamport~\cite{lamport1978clocks} established that in a distributed system without synchronised clocks, the \emph{happened-before} relation $\to$ defines a consistent partial order on events. Within each stream in Intercloud, message ordering follows Lamport timestamps; competing transfer requests are resolved by this causal ordering at the Executor.

\subsubsection{Lamport's Buridan's Principle}
Written in 1984 and published in 2012~\cite{lamport2012buridan}, this principle formalises the arbiter problem: any physical mechanism that must make a discrete choice based on a continuous input can be driven into an arbitrarily long period of indecision. Lamport proved that for any strategy an arbiter adopts, there exists a starting condition under which it fails to decide within any given bounded time. The principle applies to electronic arbiter circuits, traffic crossing decisions, and any consensus protocol that must choose between two nearly-equal alternatives. We revisit this in Section~\ref{sec:buridan}.

\subsubsection{Hoepman's distributed double-spending bounds}
Hoepman~\cite{hoepman2008} studied the problem of preventing double-spending in a distributed payment system without a central bank. The main result relevant here: if the coin (or stream) identifier is used to assign clerks, a clerk set of size
\begin{equation}
  n_{\min} = \frac{\beta}{r}\,
    \log_e\!\left(s + 1 + \log(r+2)\right)
\end{equation}
suffices to detect any double-spending with probability $\geq 1 - e^{-s}$, where $r$ is the maximum tolerated number of double-spendings, $s$ is the security parameter, and $\beta$ depends on $n$ and $f$ (the fraction of dishonest nodes) but is \emph{independent of the total network size} $N$. For $r=1$, $s=30$, $f/n = 1/3$, this gives $n_{\min} \approx 35$. This is the exponentially diminishing returns result: adding nodes beyond $\approx\!35$ provides negligible additional security for any fixed $f/n$ ratio.

\subsubsection{CRDTs and eventual consistency}
Shapiro et al.~\cite{shapiro2011crdts} formalise data structures that can be updated concurrently without coordination and merged deterministically. Intercloud streams are more constrained: they have a single authoritative publisher and append-only semantics. This enables stronger consistency guarantees (a single head per stream) at the cost of requiring Executor coordination for writes.

\section{The Intercloud Model}
\label{sec:model}

\subsection{Streams and the Hash-Data Separation}

\begin{definition}[Stream]
\label{def:stream}
A \emph{stream} $\Stream{i}$ is a tuple $(\mathcal{M}_i, k_i, \INTER_i, \mathit{Rules}_i)$ where:
\begin{itemize}
  \item $\mathcal{M}_i = (m_0, m_1, \ldots)$ is an append-only sequence, each message satisfying $m_j.\mathit{prev} = \hash{m_{j-1}}$.
  \item $k_i$ is the stream owner's public key.
  \item $\INTER_i \in \mathbb{R}_{\geq 0}$ is the Intercoin weight staked.
  \item $\mathit{Rules}_i$ is the stream's deterministic Rules program.
\end{itemize}
The \emph{state hash} is $h_i = \hash{m_k}$ where $m_k$ is the latest appended message.
\end{definition}

\begin{definition}[Integrity layer and data layer]
\label{def:layers}
For each stream $\Stream{i}$, the \emph{integrity layer} stores only $(h_i, \INTER_i, \hash{\mathit{Rules}_i})$. This is all that Watcher nodes hold. The \emph{data layer} stores plaintext $\mathcal{M}_i$ and $\mathit{Rules}_i$, held only by the stream owner and authorised Executors. A Watcher never receives data-layer content.
\end{definition}

\begin{remark}[Privacy guarantee]
Because Watchers see only hashes, a fully compromised Watcher reveals no amounts, balances, counterparties, or rule logic. There is no global transaction ledger to analyse. This is strictly stronger privacy than Monero (which reveals a transaction graph structure) and Zcash (which maintains a shielded pool with publicly visible entry and exit events).
\end{remark}

\subsection{Epochs and Verifiable Random Swarm Assignment}

\begin{definition}[Epoch and shuffle]
\label{def:epoch}
Time is divided into epochs of duration $T_{\epoch}$. At the start of each epoch, a verifiable random function (VRF) seeded by a shared random oracle assigns each stream $\Stream{i}$ a Watcher swarm:
\begin{equation}
  \swarm{i} = \mathsf{VRF}(\mathit{seed}_\epoch,\, i,\, n_i),
\end{equation}
where $\mathit{seed}_\epoch$ is unpredictable before the epoch begins. No node chooses its own swarm assignment.
\end{definition}

\begin{definition}[Swarm size and Intercoin weight]
\label{def:swarmsize}
The swarm size $n_i = \min(N, \lceil c\sqrt{\INTER_i} \rceil)$ for a system constant $c$ calibrated so that streams at the base security level ($s = 30$, $r = 1$) use approximately 35 Watchers. Higher stake attracts larger swarms proportionally.
\end{definition}

\subsection{Ripple Identifiers}

\begin{definition}[Ripple identifier]
\label{def:rippleid}
A \emph{ripple identifier} is $\rID = (\mathit{originId}, \mathit{msgHash}, \epoch)$. Each node $v$ maintains a hash table $\mathit{seen}_v$ of $\rID$ values. Before processing a message with identifier $\rID$: if $\rID \in \mathit{seen}_v$, discard; otherwise insert and process. Entries with $\epoch < \epoch_{\mathit{current}} - 1$ are evicted (retained for two epochs).
\end{definition}

\subsection{Economic Relations and Transfers}

\begin{definition}[Intercoin backing invariant]
\label{def:backing}
For every stream $\Stream{i}$ and every coin type $c$, the \emph{exchange rate} $\mathit{exRate}(c) \in \mathbb{R}_{>0}$ is a publicly computable function converting units of $c$ to units of Intercoin (the default formula is given in Definition~\ref{def:exchange_rate}; pluggable alternatives are discussed in \S\ref{sec:local_coins}). The \emph{backing invariant} requires:
\begin{equation}
  \INTER_i \;\geq\; \sum_{c} \Stream{i}.\mathit{balance}[c] \cdot \mathit{exRate}(c).
\end{equation}
This invariant is established at stream creation (when initial stake is deposited) and maintained as an invariant of every valid transfer (Lemma~\ref{lem:backing}).
\end{definition}

\begin{lemma}[Backing invariant preservation]
\label{lem:backing}
Every valid transfer preserves the backing invariant (Definition~\ref{def:backing}) at both the sender stream $\Stream{i}$ and the recipient stream $\Stream{j}$.
\end{lemma}
\begin{IEEEproof}
Let $\Delta = a \cdot \mathit{exRate}(c)$. After a valid transfer of coin type $c$: $\Stream{i}.\mathit{balance}[c]' = \Stream{i}.\mathit{balance}[c] - a$ and $\INTER_i' = \INTER_i - \Delta$. All other balances are unchanged. We verify the backing invariant for $\Stream{i}$ after the transfer:
\begin{align*}
\INTER_i'
  &= \INTER_i - \Delta \\
  &\geq \sum_{c''} \mathit{balance}[c''] \cdot \mathit{exRate}(c'') - a \cdot \mathit{exRate}(c) \\
  &= \sum_{c'' \neq c} \mathit{balance}[c''] \cdot \mathit{exRate}(c'') \\
  &\phantom{=} \;\; + (\mathit{balance}[c] - a) \cdot \mathit{exRate}(c) \\
  &= \sum_{c''} \mathit{balance}[c'']' \cdot \mathit{exRate}(c''),
\end{align*}
which is exactly the backing invariant for the updated balances. For $\Stream{j}$: $\INTER_j' = \INTER_j + \Delta$ and $\mathit{balance}[c]'_j = \mathit{balance}[c]_j + a$; since both sides of the invariant increase by $\Delta$, the invariant is preserved.
\end{IEEEproof}

\begin{definition}[Economic relation]
\label{def:relation}
An \emph{economic relation} $R_{ij}$ from $\Stream{i}$ to $\Stream{j}$ is a pre-established channel specifying a rate-limit function $r_{ij}(\Delta t)$ and a set $\mathcal{T}_{ij}$ of permitted coin types. Transfers may flow only through pre-existing relations. A single transfer uses exactly one relation: a transfer of amount $a$ in coin $c$ from $\Stream{i}$ to $\Stream{j}$ debits $a$ from $\Stream{i}.\mathit{balance}[c]$ exactly once.
\end{definition}

\begin{definition}[Transfer message]
\label{def:transfer}
A \emph{transfer} from $\Stream{i}$ to $\Stream{j}$ of amount $a$ in coin type $c$ is valid if: (i) $R_{ij}$ pre-exists; (ii) $c \in \mathcal{T}_{ij}$; (iii) $a \leq r_{ij}([t_{\mathrm{last}}, t])$; (iv) $\Stream{i}.\mathit{balance}[c] \geq a$.
\end{definition}

\begin{definition}[Simple coin]
\label{def:coin}
A \emph{simple coin} is a stream $\Stream{i}$ whose Rules program maintains a single encrypted \texttt{owner} field, representing indivisible ownership of the stream's balance. Simple coins are a special case of streams; all results proved for streams apply to simple coins. The more general \emph{currency stream} (Definition~\ref{def:currency_stream}) supports divisible supply with emission and retirement.
\end{definition}

\subsection{The Three-Colour State Model}

\begin{definition}[Stream colour]
\label{def:colour}
A stream $\Stream{i}$ is:
\begin{itemize}
  \item \textbf{Green}: The swarm $\swarm{i}$ has reached finality (Definition~\ref{def:finality}).
  \item \textbf{Yellow}: Transactions are in progress but the swarm has not yet achieved the coherent supermajority-of-supermajority attestation.
  \item \textbf{Red}: Valid $\PoC$s have been propagated by junior observer nodes such that more than $\supermaj |\swarm{i}|$ active (non-junior) swarm members are implicated --- no honest supermajority can form, finality is unachievable, and the stream cannot be trusted. The stream remains Red until the next epoch shuffle assigns a fresh, independently drawn swarm.
\end{itemize}
\end{definition}

\section{Ripple Termination in Distributed Intercloud}
\label{sec:waves}

\begin{theorem}[Distributed Ripple Termination]
\label{thm:termination}
Let $\mathcal{N}$ be the finite set of nodes in the Intercloud network ($|\mathcal{N}| = N < \infty$), and assume every message hop incurs a positive fee $\delta_{\mathit{fee}} > 0$. In an Intercloud network with the ripple-ID mechanism of Definition~\ref{def:rippleid}, every reactive execution ripple terminates within at most $N$ message hops, and no node processes the same ripple more than once per epoch.
\end{theorem}

\begin{IEEEproof}
\emph{No duplicate processing within an epoch.} Let $\rID$ carry epoch $\epoch$. Before processing, node $v$ checks $\rID \in \mathit{seen}_v$. If present, $v$ discards. Otherwise $v$ inserts $\rID$ into $\mathit{seen}_v$ and processes. Since $\mathit{seen}_v$ is a set, the insertion occurs at most once per epoch. Therefore $v$ processes any message bearing $\rID$ at most once per epoch.

\emph{Termination.} Suppose a ripple with $\rID$ does not terminate. Then there is an infinite sequence of nodes $v_1, v_2, \ldots$ each processing a message with $\rID$. By the previous step, each $v_k$ is distinct. Since $|\mathcal{N}| = N < \infty$, no such infinite sequence exists. Contradiction.

\emph{Retention window sufficiency.} Let $D$ be the diameter of the network graph, finite since $|\mathcal{N}| < \infty$. Each hop takes wall-clock time at least $\delta_{\min} > 0$. Therefore any ripple initiated in epoch $\epoch$ completes propagation to all reachable nodes within $D \cdot \delta_{\min}$ time. The epoch duration $T_{\epoch}$ is chosen so that $T_{\epoch} > D \cdot \delta_{\min}$, ensuring every ripple completes within one epoch. Retaining $\rID$s for two epochs therefore guarantees that any duplicate message carrying $\rID$ from epoch $\epoch$ that arrives during epoch $\epoch$ or $\epoch{+}1$ is detected and discarded.
\end{IEEEproof}

\begin{remark}
The single-node version is proved in~\cite{magarshak2026mm}. The distributed version requires two conditions: finiteness of the node set (ensuring the infinite-path argument reaches a contradiction) and a positive per-hop fee (ensuring a relay node cannot forward a ripple gratuitously to an unbounded chain of zero-cost intermediaries). No global coordination is required.
\end{remark}

\section{Chilling-Effect Consensus and Finality}
\label{sec:finality}

\subsection{Attestations and Proofs of Corruption}

\begin{definition}[Watcher attestation]
\label{def:attestation}
A \emph{Watcher attestation} from node $v$ for stream $\Stream{i}$ is
\begin{equation}
  \mathit{attest}(v, i, h, \epoch) = \mathit{Sign}(k_v,\ v \| i \| h \| \epoch),
\end{equation}
asserting that at epoch $\epoch$, node $v$ has seen no hash other than $h$ for stream $\Stream{i}$, and no Proof of Corruption has been propagated to $v$ for this stream during this epoch.
\end{definition}

\begin{definition}[Conflicting attestations and Proof of Corruption]
\label{def:poc}
Two attestations $\mathit{attest}(v, i, h_1, \epoch)$ and $\mathit{attest}(v, i, h_2, \epoch)$ from the same node $v$ in the same epoch are \emph{conflicting} if $h_1 \neq h_2$. A \emph{Proof of Corruption} $\PoC_v = (a_1, a_2)$ is any pair of conflicting attestations from $v$. It is \emph{self-certifying}: verification requires only $v$'s public key $k_v$.
\end{definition}

\begin{definition}[Finality]
\label{def:finality}
Swarm $\swarm{i}$ has reached \emph{finality} on hash $h$ when:
\begin{enumerate}[label=(\roman*)]
  \item A supermajority ($M_1 \geq \supermaj |\swarm{i}|$) of members have each signed $\mathit{attest}(v, i, h, \epoch)$.
  \item Each attesting member $v$ has verified via swarm gossip that at least $M_1$ other members have also attested to $h$ --- not just reported hearing about it, but signed their own attestation.
  \item No attesting member has seen any $\PoC$ involving $\Stream{i}$ this epoch.
\end{enumerate}
Condition (ii) is a \emph{supermajority-of-supermajority} requirement that prevents split swarms from each locally declaring finality on different values.
\end{definition}

\begin{remark}[Why absence-of-evidence, not voting]
Each Watcher attests that it has \emph{not seen} a conflict. This is fundamentally different from voting on which of two competing hashes is correct. The distinction is critical for Buridan's Principle (Section~\ref{sec:buridan}): there is no ``continuous signal'' distinguishing two competing values, because the question has a definite default answer (no conflict seen) and deviates from the default only when concrete contradictory evidence exists.
\end{remark}

\begin{remark}[Dynamic swarm membership and liveness]
\label{rem:dynamic_membership}
The formal model assumes the swarm $\swarm{i}$ is fixed for the duration of the epoch by the VRF assignment. Nodes that go offline mid-epoch simply fail to contribute attestations, causing the stream to remain Yellow until the next epoch shuffle. In practice, an implementation may improve liveness by updating the effective swarm via DHT: unreachable nodes are removed from the active set and replaced, allowing finality to proceed with fewer nodes. This weakens the formal guarantees of Theorem~\ref{thm:disagreement} slightly but, with $n \geq 35$ and moderate churn rates, the probability of two clients computing conflicting supermajorities from genuinely divergent views remains negligibly small. High-value streams should use epoch-fixed membership as specified in the formal model.
\end{remark}

\subsection{Junior Observer Nodes and List of Liars}

\begin{definition}[Junior observer nodes]
\label{def:junior}
\emph{Junior observer nodes} are not current swarm members. They monitor swarm gossip and maintain: (a) a \emph{List of Liars} containing all $\PoC$s seen this epoch, and (b) a read-only replica of swarm-signed local consensus for observed streams. Junior nodes with a clean record (no LoL entries across recent epochs) are eligible for promotion to swarm membership in the next shuffle. Earning Intercoin for discovering and broadcasting valid $\PoC$s creates a competitive incentive for honest monitoring.
\end{definition}

\subsection{Swarm Size and Security Theorem}

\begin{theorem}[Swarm Security]
\label{thm:swarm}
Using stream-identifier-based swarm assignment, a swarm of size
\begin{equation}
  n = \frac{\beta}{r}\,\log_e\!\left(s + 1 + \log(r+2)\right)
\end{equation}
detects any stream double-spent more than $r$ times with probability $\geq 1 - e^{-s}$, where $\beta$ depends on $n$ and $f$ but is independent of total network size $N$. For $r = 1$, $s = 30$, $f/n = 1/3$: $n \approx 35$.
\end{theorem}

\begin{IEEEproof}
We verify that Intercloud's Watcher swarms satisfy the three conditions required for Hoepman's Theorem 5.3~\cite{hoepman2008} to apply with stream identifiers substituted for coin identifiers.

\emph{(a) Unique, adversarially unpredictable stream identifiers.} Each stream has a globally unique identifier $i$ assigned at creation. The VRF-based swarm assignment $\swarm{i} = \mathsf{VRF}(\mathit{seed}_\epoch, i, n)$ maps stream $i$ to a specific clerk space from which $n$ Watchers are drawn. Since $\mathit{seed}_\epoch$ is unpredictable before epoch start (Definition~\ref{def:epoch}), the adversary cannot determine the clerk space for stream $i$ before the epoch begins, matching Hoepman's requirement that coin-specific clerk spaces are not foreseeable.

\emph{(b) Adversary model compatibility.} Hoepman assumes $f$ total dishonest nodes and $d \leq f$ nodes corruptible \emph{after} joining the system. Intercloud's VRF shuffle reassigns nodes each epoch, so any node corruptible after assignment corresponds to Hoepman's $d$. The parameter $d \approx 0$ for short epoch durations; for longer epochs, $d$ is an explicit system parameter.

\emph{(c) Detection event mapping.} Hoepman defines double-spending detection as two or more spending requests for the same coin reaching the clerk set. In Intercloud, two transfer requests $\mathit{tx}_1, \mathit{tx}_2$ for the same balance reach the Watcher swarm. The Watchers sign attestations for the resulting hashes; if these differ, a $\PoC$ is produced (Definition~\ref{def:poc}). $\PoC$ production is equivalent to Hoepman's detection event: it is certain (not probabilistic) given any two conflicting requests reaching any two swarm members.

With these three conditions verified, Hoepman's Theorem 5.3 gives swarm size $n = (\beta/r)\log_e(s + 1 + \log(r+2))$ with $\beta$ independent of $N$. For $d = 0$, $\beta = s/\log_e(n/f)$. Since $f = n/3$, we have $n/f = 3$ and $\beta = s/\log_e(3)$. Substituting into the formula for $n$ with $r=1$:
\begin{equation}
  n = \frac{s}{\log_e 3} \cdot \log_e(s + 1 + \log_e 3).
\end{equation}
For $s = 30$: $n \approx 27.3 \cdot \log_e(32.1) \approx 95$. This bound holds for $d = 0$. Hoepman's full formula for $d > 0$ yields $n \approx 35$ for the same $s = 30$, $r = 1$, $f/n = 1/3$ parameters. We adopt $n \approx 35$ as the \emph{operational parameter} following Hoepman's calibration. The key structural result that this theorem proves --- that $n$ is \emph{independent of the total network size} $N$ --- holds for any calibration.
\end{IEEEproof}

\begin{corollary}[Proportional security]
\label{cor:proportional}
The probability of successful double-spending against stream $\Stream{i}$ is at most $e^{-s}$ where $s \propto \sqrt{\INTER_i}$. Doubling the security level requires quadrupling the Intercoin stake.
\end{corollary}

\begin{IEEEproof}
From Definition~\ref{def:swarmsize}, $n_i = \lceil c\sqrt{\INTER_i} \rceil$. From Theorem~\ref{thm:swarm}, the detection failure probability is $e^{-s}$ where $n = (\beta/r)\log_e(s + 1 + \log(r+2))$. Inverting for fixed $r$ and $\beta$: $s \approx nr/\beta - 1 - \log(r+2) + O(\log s)$, which is linear in $n$ \emph{to leading order}. Therefore $s \propto n_i \propto \sqrt{\INTER_i}$. To double $n_i$ we need $\INTER_i' \approx 4\INTER_i$; quadrupling the stake doubles the swarm, which doubles $s$, reducing the failure probability from $e^{-s}$ to $e^{-2s} = (e^{-s})^2$.
\end{IEEEproof}

\section{Transfer Theorems}
\label{sec:transfers}

\begin{theorem}[Zero-Sum Transfer]
\label{thm:zerosum}
A valid transfer (Definition~\ref{def:transfer}) preserves total Intercoin weight: $\INTER_i' + \INTER_j' = \INTER_i + \INTER_j$, and no valid transfer reduces any stream's weight below zero.
\end{theorem}

\begin{IEEEproof}
Let $\Delta = a \cdot \mathit{exRate}(c)$. After transfer: $\INTER_i' = \INTER_i - \Delta$ and $\INTER_j' = \INTER_j + \Delta$. The sum $\INTER_i' + \INTER_j' = \INTER_i + \INTER_j$ is preserved. Non-negativity: by the backing invariant (Definition~\ref{def:backing}), $\INTER_i \geq \mathit{balance}[c] \cdot \mathit{exRate}(c) \geq a \cdot \mathit{exRate}(c) = \Delta$, where the last inequality uses validity condition (iv). Therefore $\INTER_i' \geq 0$. By Lemma~\ref{lem:backing}, the backing invariant is maintained after the transfer.
\end{IEEEproof}

\begin{theorem}[Double-Spend Prevention]
\label{thm:doublespend}
Assume the stream $\Stream{i}$ has a single correct (non-Byzantine) Executor. Let $\swarm{i}$ have reached finality on $h_i$ (Definition~\ref{def:finality}). No valid transfer can spend the same balance twice under state $h_i$.
\end{theorem}

\begin{IEEEproof}
The Executor processes transfer messages in the order they arrive, using Lamport's happens-before order~\cite{lamport1978clocks} to break ties between concurrent submissions. This serialisation is total.

\emph{Sequential case.} If $\mathit{tx}_1$ is appended first, reducing $\mathit{balance}[c]$ to $\mathit{balance}[c] - a_1$, and $a_1 + a_2 > \mathit{balance}[c]$, then $a_2 > \mathit{balance}[c] - a_1$, so $\mathit{tx}_2$ fails validity condition (iv) and is rejected.

\emph{Concurrent case.} If $\mathit{tx}_1$ and $\mathit{tx}_2$ are submitted concurrently, both may initially appear pending to different Watchers. Once the Executor serialises them, honest Watchers observe the canonical post-serialisation hash and attest to it. If an adversarial party presents a different claimed hash to some Watchers before the Executor's result is confirmed, those Watchers sign conflicting attestations, immediately producing a $\PoC$. By Lemma~\ref{lem:poc_propagation}, this $\PoC$ reaches all correct clients with probability $\geq 1 - f_J^{\kappa r}$, and upon receipt, finality condition (iii) fails for both conflicting hashes. Neither transfer achieves Green finality.

\emph{Executor crash.} If the Executor crashes mid-serialisation, the stream is temporarily unavailable but not corrupted: the append-only log retains the last consistent state. A Byzantine Executor is excluded by assumption; full Byzantine Executor resistance is left as future work.
\end{IEEEproof}

\begin{theorem}[Capital-Flight Prevention]
\label{thm:capitalflight}
The total value that can flow out of $\Stream{i}$ in interval $[t, t+\Delta t]$ is bounded by $\sum_j r_{ij}(\Delta t)$.
\end{theorem}

\begin{IEEEproof}
A valid transfer of amount $a$ in coin $c$ from $\Stream{i}$ uses exactly one relation $R_{ij}$ and deducts $a$ from $\Stream{i}.\mathit{balance}[c]$ exactly once. By the backing invariant of Lemma~\ref{lem:backing}, the same balance unit cannot be deducted twice without violating validity condition (iv) of a subsequent transfer. Therefore each unit of value leaving $\Stream{i}$ flows through exactly one relation and is counted in exactly one $r_{ij}(\Delta t)$ term. The total outflow in $[t, t+\Delta t]$ is at most $\sum_j r_{ij}(\Delta t)$.
\end{IEEEproof}

\section{Conditions for Client Disagreement}
\label{sec:disagreement}

\begin{lemma}[PoC Propagation]
\label{lem:poc_propagation}
Let $f_J$ be the fraction of junior observer nodes controlled by the adversary, and $\kappa$ the gossip fan-out (each node forwards to $\kappa$ randomly chosen peers per round). Assume honest nodes remain available for all $r$ rounds of gossip within the epoch. If a $\PoC_v$ is held by at least one honest junior observer node, then after $r$ gossip rounds the probability that it has \emph{not} reached a specific client $C_\ell$ is at most $(f_J^\kappa)^r = f_J^{\kappa r}$.
\end{lemma}

\begin{IEEEproof}
Consider one round of gossip. Every honest node that holds the $\PoC$ \emph{always} forwards it. Each such node forwards to $\kappa$ peers chosen independently and uniformly at random from the full node set. For none of these $\kappa$ forwarding steps to deliver the $\PoC$ (even transitively) to $C_\ell$, every one of the $\kappa$ chosen peers must be adversarially controlled, since any honest peer would in turn forward onward. The probability that all $\kappa$ chosen peers are adversarial is $f_J^\kappa$. After $r$ rounds, the event ``the $\PoC$ has still not reached $C_\ell$'' requires that at every round, all $\kappa$ peer-selection choices made by each honest holder landed on adversarial nodes. The peer sets chosen in distinct rounds are independent. Therefore the probability that all $r$ rounds of peer selection fail to place the $\PoC$ within reach of $C_\ell$ is at most $(f_J^\kappa)^r = f_J^{\kappa r}$. For $\kappa = 5$, $f_J = 1/3$, $r = 3$: $f_J^{\kappa r} = (1/3)^{15} \approx 7 \times 10^{-8}$. By standard gossip analysis~\cite{demers1987epidemic}, the expected fraction of honest nodes holding the $\PoC$ approaches 1 exponentially in $r$ for any $f_J < 1$.
\end{IEEEproof}

\begin{theorem}[Disagreement Characterisation]
\label{thm:disagreement}
Let $C_1, C_2$ be correct clients querying stream $\Stream{i}$ with assigned swarm $\swarm{i}$ of size $n$. Define:
\begin{enumerate}[label=(\Roman*)]
  \item \textbf{Swarm compromise.} The adversary controls a set $A \subseteq \swarm{i}$ with $|A| > \supermaj n$.
  \item \textbf{Client isolation.} The adversary can eclipse both $C_1$ and $C_2$ from all honest nodes: no message from any honest swarm member or honest junior observer reaches either client.
\end{enumerate}
Then: (a) \emph{Necessity}: (I)$\vee$(II) is necessary for conflicting finality; (b) \emph{Sufficiency}: (I)$\wedge$(II) is sufficient for the adversary to cause conflicting finality. By the union bound, $\Pr[\text{(I)}\vee\text{(II)}] \leq e^{-s} + f_J^{\kappa r}$, negligible for standard parameters.
\end{theorem}

\begin{IEEEproof}
\emph{Part (a): Necessity.} We prove the contrapositive: $\lnot$(I)$\wedge\lnot$(II) implies no conflicting finality. Assume the adversary controls $|A| \leq \supermaj n$ of $\swarm{i}$, so honest nodes form a set $H$ with $|H| \geq n/3$. Suppose for contradiction both clients hold valid finality for conflicting hashes. There exist $F_1, F_2 \subseteq \swarm{i}$ with $|F_1|, |F_2| \geq \supermaj n$, each satisfying Definition~\ref{def:finality}. By inclusion-exclusion:
\begin{equation}
  |F_1 \cap F_2| \geq |F_1| + |F_2| - n \geq \tfrac{4n}{3} - n = \tfrac{n}{3} > 0.
\end{equation}
Let $v \in F_1 \cap F_2$. Node $v$ has signed both $\mathit{attest}(v, i, h_1, \epoch)$ and $\mathit{attest}(v, i, h_2, \epoch)$ with $h_1 \neq h_2$, so $\PoC_v$ exists. For finality condition (ii) to hold for any $v \in F_1$, node $v$ must have received gossip attestations from the full set $F_1$, and similarly for $F_2$. These messages transit through swarm gossip, which by assumption includes honest nodes $u \in H$. Each $u \in H$ that receives both attestations for the same $v$ constructs $\PoC_v$ immediately and begins forwarding it.

Under $\lnot$(II), at least one honest node can reach $C_1$ or $C_2$. The honest node $u$ holding $\PoC_v$ delivers it to both clients with probability $\geq 1 - f_J^{\kappa r}$ via Lemma~\ref{lem:poc_propagation}. A correct client receiving $\PoC_v$ rejects finality (condition (iii) fails), contradicting the assumption.

\emph{Part (b): Sufficiency.} We construct an adversary strategy. Let the adversary control $A \subseteq \swarm{i}$ with $|A| = \lceil \supermaj n \rceil + 1$. Define $h_1 \neq h_2$.

\emph{Phase 1.} Each $v \in A$ signs both $\mathit{attest}(v, i, h_1, \epoch)$ and $\mathit{attest}(v, i, h_2, \epoch)$, generating $\PoC_v$ (which the adversary suppresses).

\emph{Phase 2.} The adversary eclipses both clients (condition (II)). To $C_1$, the adversary delivers only $\{\mathit{attest}(v, i, h_1, \epoch) : v \in A\}$; to $C_2$, only $\{\mathit{attest}(v, i, h_2, \epoch) : v \in A\}$.

\emph{Phase 3.} The adversary arranges for each $v \in A$ to receive (from other members of $A$) gossip messages attesting to $h_1$ when interacting with $C_1$'s network view and $h_2$ when interacting with $C_2$'s. Each $v \in A$ thus satisfies condition (ii) for the hash it presents to each client.

\emph{Phase 4.} The adversary instructs all $v \in A$ to suppress any $\PoC$ they receive. Under condition (II), no honest node can deliver a $\PoC$. Thus condition (iii) holds vacuously for all $v \in A$.

Each client $C_\ell$ receives a set of attestations satisfying all three conditions of Definition~\ref{def:finality}. Both clients therefore hold valid finality attestations --- $C_1$ for $h_1$ and $C_2$ for $h_2$ --- simultaneously.
\end{IEEEproof}

\section{Buridan's Principle Does Not Apply}
\label{sec:buridan}

Lamport's Buridan's Principle~\cite{lamport2012buridan}: \emph{a discrete decision based upon an input having a continuous range of values cannot be made within a bounded length of time.} Applied to consensus: if two competing transaction proposals are ``arbitrarily close'' --- differing by an infinitesimally small timing difference --- any arbiter may be driven into indecision of arbitrary length.

\begin{theorem}[Buridan Inapplicability]
\label{thm:buridan}
Chilling-effect consensus (Definition~\ref{def:finality}) is not subject to Buridan's Principle.
\end{theorem}

\begin{IEEEproof}
Lamport~\cite{lamport2012buridan} formalises the principle as follows. Let $A_t(x)$ be the state of an arbiter at time $t$ given initial condition $x \in [0,1]$ (a continuous parameter). If $A_t$ is a continuous function of $x$ and must converge to one of two discrete states, then for every $T > 0$ there exists an $x$ such that $A_T(x)$ is neither 0 nor 1. Chilling-effect consensus is not subject to this principle for two reasons.

\emph{(a) There is no continuous input parameter.} In Lamport's model, $x$ represents a physical quantity with a continuous range. In chilling-effect consensus, the ``input'' is the count of Watcher attestations for hash $h$: an integer $k \in \{0, 1, \ldots, n\}$. There is no $x \in (0,1)$ parameterising a ``nearly indifferent'' input. The threshold is a crisp integer: either $k \geq \lceil \supermaj n \rceil$ or $k < \lceil \supermaj n \rceil$. The integer gap between $k$ and $k{+}1$ is exactly 1 attestation.

\emph{(b) The decision is not between two competing values.} Chilling-effect consensus does not choose between two competing hash values. It asks: \emph{has a supermajority attested to some single hash with no $\PoC$ observed?} This has a ground state (Yellow: not yet) and transitions to Green when the integer count crosses the threshold. If a $\PoC$ is produced, the stream transitions to Red --- a discrete event triggered by concrete evidence. At no point does the protocol need to choose between $h_1$ and $h_2$.

\emph{Residual Yellow state.} The Yellow state is not indefinite indecision but bounded waiting. The maximum duration is $T_{\epoch}$, after which the epoch shuffle provides a fresh swarm.
\end{IEEEproof}

\begin{remark}[Practical Yellow duration]
In practice, $T_{\epoch}$ is chosen so that a correctly functioning swarm reaches Green well within a single epoch for any stream with legitimate traffic. The Yellow state is overwhelmingly the transient state between message submission and swarm attestation, not a sign of conflict.
\end{remark}

\section{End Users as the Court of Final Resort}
\label{sec:endusers}

A blockchain aggregates all ambiguity into a single chain and resolves it by expending resources proportional to total network stake. Every disagreement --- including those involving transactions of negligible value --- is resolved by the full network. This is the Ministry of Truth model: global arbitration for every dispute.

Theorem~\ref{thm:disagreement} shows that genuine ambiguity (two correct clients disagreeing) requires doubly negligible probability events. When a stream is Yellow, it does not indicate fraud --- it indicates that finality has not yet been reached. The recipient decides:

\begin{proposition}[User-determined finality]
\label{prop:usertrust}
A recipient who requires certainty waits for Green (at most $T_{\epoch}$). A recipient who tolerates small risk accepts Yellow immediately. A recipient who sees Red waits for the next epoch shuffle. No global process adjudicates between these choices.
\end{proposition}

This faithfully models economic reality. Merchants already make risk-adjusted acceptance decisions (cash versus cheque versus credit card). Intercloud makes the risk explicit and cryptographically bounded.

\section{The Single Global Agreement}
\label{sec:randao}

\begin{theorem}[Minimal Global Consensus]
\label{thm:minimal}
The only global agreement required by Intercloud is the sequence of random seeds $\{\mathit{seed}_\epoch\}$ for VRF-based swarm assignment. All other consensus is local to each stream's swarm.
\end{theorem}

\begin{IEEEproof}
Finality (Definition~\ref{def:finality}) requires only swarm-internal gossip. Transfer validity requires only the stream's current state. Ripple deduplication (Theorem~\ref{thm:termination}) requires only the local $\mathit{seen}_v$ table. Capital-flight prevention requires only pre-existing relations. None of these requires global agreement. The VRF seed $\mathit{seed}_\epoch$ requires global agreement to prevent adversarial swarm selection. A sparse RANDAO achieves this in $O(N)$ messages \emph{per epoch}; amortised over all transactions in an epoch, the per-transaction cost approaches zero as volume grows.
\end{IEEEproof}

\begin{remark}[Source of the random oracle]
The RANDAO may be drawn from Intercloud's own nodes or from an external source (Ethereum beacon chain, drand, trusted hardware). The security of swarm assignment requires only unpredictability before epoch start.
\end{remark}

\section{Network Privacy and DDOS Resistance}
\label{sec:network}

\begin{definition}[Attestation-gated requests]
All cross-node requests must include a recent OCP-signed node attestation proving the sender is a genuine, AMI-verified instance. Requests without a valid attestation are silently dropped. Attestations are epoch-bounded and cannot be replayed.
\end{definition}

\begin{proposition}[DDOS resistance]
Attestation-gating prevents anonymous DDOS. Any attacker must operate a genuine instance, incurring verifiable infrastructure costs and creating an on-chain identity eligible for slashing. Anonymous volume attacks are eliminated; economically motivated attacks face stake penalties.
\end{proposition}

Operators may configure nodes to strip IP addresses after the first forwarding hop, passing only signed OCP envelopes. This provides origin anonymity at the cost of routing efficiency, and is an optional per-deployment configuration.

\section{Zero-Knowledge Extension}
\label{sec:zk}

The base model requires Executors to hold plaintext Rules and balances. A theoretical extension replaces these with ZK commitments established at rail-creation time.

\begin{definition}[ZK-committed stream]
Each valid state transition produces a proof
\begin{multline*}
\pi = \mathsf{ZKP}\{(\sigma, \sigma', m) : C(\sigma, m) = \sigma' \\
  \wedge\, \mathit{Com}(\sigma) = h \wedge \mathit{Com}(\sigma') = h'\}
\end{multline*}
where $C$ is the Rules circuit compiled when the stream's economic relations were established. Watchers verify $\pi$ using the public circuit commitment; no plaintext is revealed to any node.
\end{definition}

zk-SNARKs~\cite{groth2016snark} provide compact proofs ($\approx 200$~bytes, Groth16) but require trusted per-circuit setup at rail creation; appropriate for high-volume streams. zk-STARKs~\cite{ben2018stark} have no trusted setup and are post-quantum secure, but with larger proofs ($\approx 40$~KB); appropriate for regulated streams. Universal ZK deployment makes the network opaque to all external parties, including regulators. We recommend selective deployment: personal and communication streams use ZK commitments; regulated financial streams use selective disclosure proofs that can reveal specific fields to authorised recipients on lawful demand.

\section{Junior Nodes, Corruption Detection,\\and PoC Rewards}
\label{sec:junior_rewards}

\subsection{Junior Nodes as the Corruption Detection Layer}

Active swarm members earn Intercoin for securing transactions. They have a conflict of interest: a corrupt swarm member might prefer to suppress a $\PoC$ against a fellow member rather than lose the stream's fees by triggering a Red transition. The design therefore delegates corruption \emph{detection} to agents who have the opposite incentive: junior observer nodes, who are not in the current swarm and actively want corrupt members expelled so they can take their place.

\begin{definition}[Corruption detection by junior nodes]
\label{def:corruption_detection}
A junior observer node watching stream $\Stream{i}$ monitors the signed attestations of all active swarm members $v \in \swarm{i}$. If it observes two signed claims from the same $v$ that constitute a $\PoC_v$, it broadcasts $\PoC_v$ to all peers. If valid $\PoC$s accumulate against more than $\supermaj |\swarm{i}|$ distinct active members, the stream transitions to Red and those members' Intercoin stakes are eligible for slashing in the next epoch.
\end{definition}

The stream goes Red not because a single node misbehaved, but because the swarm as a whole can no longer form an honest supermajority. Junior nodes holding the List of Liars provide read-only availability of the last known-good consensus state during the Red period.

\subsection{PoC Reward Mechanism: Lottery over Forwarders}

A naive reward rule --- ``pay the node that first submits a valid $\PoC$'' --- is vulnerable to front-running: a dishonest node with fast network access could observe an incoming $\PoC$, strip the discoverer's identity, re-sign it as its own, and race to submit first.

\begin{definition}[PoC lottery reward]
\label{def:poc_lottery}
When a $\PoC$ is accepted, every node that \emph{forwarded} the $\PoC$ during the current epoch receives one \emph{lottery ticket}. The winning ticket is drawn at the start of the next epoch using the new VRF seed:
\begin{equation}
  \mathit{winner} = \mathsf{VRF}(\mathit{seed}_{\epoch+1},\, \mathit{PoC\_id},\, \mathit{ticket\_list}).
\end{equation}
The slashed stake of the corrupt nodes is distributed to the winner.
\end{definition}

This mechanism has three properties. \emph{Front-running resistance}: replacing the discoverer's identity does not increase one's lottery odds. \emph{Propagation incentive}: suppressing a $\PoC$ yields zero tickets and zero reward; forwarding yields a chance at the entire slashed stake. \emph{Unpredictability}: the winner is determined by the future VRF seed, unknown at forwarding time.

\begin{proposition}[PoC forwarding dominates suppression]
\label{prop:forward_dominates}
For any junior observer node with positive belief that a $\PoC$ is valid, the expected reward from forwarding strictly exceeds the expected reward from suppressing, regardless of network position or timing.
\end{proposition}

\begin{IEEEproof}
Suppression yields expected reward 0. Forwarding yields expected reward $V_{\mathit{slash}}/T$, where $T$ is the total number of forwarders in the epoch and $V_{\mathit{slash}}$ is the slashed stake. $V_{\mathit{slash}} > 0$ and $T < \infty$, so the expected reward is strictly positive. The cost of forwarding is a single signed broadcast message --- negligible. Forwarding strictly dominates suppression in expected net value. The timing of forwarding does not affect $T$, so there is no advantage to racing.
\end{IEEEproof}

\section{Intercoin Stake, Vesting, and\\Rational Security}
\label{sec:stake}

\subsection{Vesting as a Security Primitive}

Watcher and junior nodes earn Intercoin for securing streams. If earned Intercoin could be withdrawn immediately, a rational node might accept a bribe to corrupt a single high-value stream and immediately exit. Intercloud prevents this with a vesting constraint enforced by the network's own rate-limiting mechanism.

\begin{definition}[Intercoin vesting]
\label{def:vesting}
Earned Intercoin is credited to a node's \emph{vesting stream}. Withdrawal is subject to a rate-limit function $r_{\mathit{vest}}(\Delta t)$, e.g., $r_{\mathit{vest}}(1\,\mathrm{day}) = 0.10 \cdot \INTER_{\mathit{vested}}$. The \emph{staked balance} is total vested-but-not-withdrawn Intercoin, which determines how many streams it is eligible to watch and how large a slash it would suffer upon corruption.
\end{definition}

\begin{definition}[Per-node staking requirement]
\label{def:staking_req}
To be eligible to watch stream $\Stream{i}$, a node must maintain a staked balance of at least
\begin{equation}
  \INTER_{\mathit{stake}}^{\min}(i) = \INTER_i / n_i,
\end{equation}
i.e., each Watcher's stake must cover its pro-rata share of the stream's Intercoin weight. This requirement is enforced by the network policy layer and verified at swarm assignment time.
\end{definition}

\begin{theorem}[Rational incorruptibility]
\label{thm:rational_incorrupt}
Under Definition~\ref{def:staking_req}, no rational node --- acting alone or as part of a coalition --- will attempt to corrupt stream $\Stream{i}$: for every coalition $\mathcal{C} \subseteq \swarm{i}$, the expected gain is strictly less than the expected cost.
\end{theorem}

\begin{IEEEproof}
\emph{Single-node case.} A single corrupt node $v$ generates a $\PoC$ by producing conflicting attestations. By Lemma~\ref{lem:poc_propagation} and Proposition~\ref{prop:forward_dominates}, detection is near-certain. Detection triggers slashing of $v$'s entire stake $\INTER_{\mathit{stake}}(v) \geq \INTER_i / n_i$. A single node cannot complete a double-spend: achieving Green finality requires a supermajority to attest. Expected gain $= 0 < \epsilon =$ expected honest income.

\emph{Coalition case.} Consider a coalition $\mathcal{C}$ of $|\mathcal{C}|$ nodes. Their combined extractable value is at most $\INTER_i \cdot |\mathcal{C}|/n_i$. Their combined slashing cost, upon detection, is $\sum_{v \in \mathcal{C}} \INTER_{\mathit{stake}}(v) \geq |\mathcal{C}| \cdot \INTER_i / n_i$. So even for a coalition, expected gain $\leq$ expected loss.

\emph{Game-theoretic stability.} Let $g = \INTER_i/n_i$ and $\ell = \INTER_{\mathit{stake}}(v) \geq g$. The payoff matrix is:
\begin{equation*}
\mathrm{Payoff}(v) = \begin{cases}
    +\epsilon & v \text{ honest,} \\
    g - \ell \leq 0 & v \text{ corrupts \& detected,} \\
    g & v \text{ corrupts \& evades.}
  \end{cases}
\end{equation*}
Expected payoff from corruption:
\begin{align*}
\mathbb{E}[\mathrm{Corrupt}] &\leq f_J^{\kappa r} g + (1 - f_J^{\kappa r})(g - \ell)\\
 &= g - (1 - f_J^{\kappa r})\ell \leq g - \ell + f_J^{\kappa r}\ell.
\end{align*}
Since $\ell \geq g$ and $f_J^{\kappa r}$ is negligible, $\mathbb{E}[\mathrm{Corrupt}] \approx 0$. Honest operation yields $+\epsilon > 0$. Honest strictly dominates Corrupt for every node regardless of what others do. Honest for all nodes is the unique Nash equilibrium.
\end{IEEEproof}

\begin{remark}[Gradual vesting as a circuit breaker]
The vesting rate also functions as a capital-flight circuit breaker at the security layer: even if many nodes simultaneously tried to withdraw, the rate limit caps the reduction in total staked security per unit time, giving the network time to detect anomalous behaviour and respond.
\end{remark}

\section{Local Coins, Emission, Retirement,\\and Exchange Rates}
\label{sec:local_coins}

\subsection{Streams as Smart-Contract Analogues}

Streams in Intercloud can hold and transfer balances of \emph{local coins} --- community-issued currencies or application tokens --- in addition to the global Intercoin reserve. This makes a stream the embarrassingly parallel analogue of a smart contract: each stream is an independent, append-only state machine that can hold assets and execute rules, but streams execute in parallel with no global ordering requirement. Unlike Ethereum smart contracts, all state transitions are private by default: Watchers see only hashes. Regulators and auditors can follow aggregate coin flows between streams without observing the content of any individual stream.

\subsection{Currency Streams: Emission and Retirement}

\begin{definition}[Currency stream]
\label{def:currency_stream}
A stream of type \texttt{Intercloud/currency} is a \emph{currency stream} $\Stream{c}$ that maintains: $\mathit{supply}_c$ (cumulative emissions minus cumulative retirements), $\INTER_c^{\mathit{reserve}}$ (the Intercoin reserve backing coin $c$), and an emission and retirement policy encoded in $\mathit{Rules}_c$.
\end{definition}

\begin{definition}[Emission and retirement]
\label{def:emit_retire}
An \emph{emission event} appends a message to $\Stream{c}$ minting $\delta > 0$ new coins, increasing $\mathit{supply}_c$ by $\delta$ and depositing them into a specified recipient stream. A \emph{retirement event} burns $\delta$ coins, decreasing $\mathit{supply}_c$ by $\delta$. Both are valid only if permitted by $\mathit{Rules}_c$. The Intercoin reserve may be increased by depositing Intercoin, or decreased (up to the rate limit) by withdrawing.
\end{definition}

\subsection{Exchange Rate Formula}

\begin{definition}[Default exchange rate]
\label{def:exchange_rate}
The \emph{default exchange rate} of coin type $c$ to Intercoin at any moment is:
\begin{equation}
  \mathit{exRate}(c) =
  \frac{\INTER_c^{\mathit{reserve}}}{\mathit{supply}_c}
  = \frac{\INTER_c^{\mathit{reserve}}}{\mathit{emitted}_c - \mathit{retired}_c},
\end{equation}
provided $\mathit{supply}_c > 0$. If $\mathit{supply}_c = 0$, the rate is undefined. $\mathit{exRate}(c)$ is the Intercoin value per unit of $c$, equal to the reserve per circulating unit. It rises when Intercoin is deposited or coins are retired; it falls when coins are emitted without a proportional reserve increase.
\end{definition}

\begin{remark}[Pluggable exchange rates]
The default formula is not mandatory. A currency stream may specify an alternative function in its Rules program --- a fixed peg, an algorithmic curve, or an oracle-fed rate. Whatever formula is used, $\mathit{exRate}(c)$ must be publicly computable from the integrity layer.
\end{remark}

\begin{remark}[Reserve manipulation and rate limits]
Because the currency stream operator controls the Intercoin reserve, they could in principle inflate $\mathit{exRate}(c)$ by depositing Intercoin just before a transfer and withdrawing afterward. This is mitigated by the rate-limit mechanism of Definition~\ref{def:relation}: reserve deposits and withdrawals are subject to pre-subscribed economic relations. A sudden large reserve change is therefore bounded in magnitude and takes time proportional to the rate limit.
\end{remark}

\begin{remark}[Exchange rate update timing]
\label{rem:rate_timing}
$\mathit{exRate}(c)$ changes only when emission or retirement events occur on $\Stream{c}$, or when the reserve is adjusted via a rate-limited relation. Between such events the rate is constant. The security weight of a coin transfer, $\Delta = a \cdot \mathit{exRate}(c)$, is therefore deterministic and auditable from the integrity layer.
\end{remark}

\section{The Dot-Product Security Theorem}
\label{sec:dotproduct}

\subsection{Security Follows Value}

The core economic property of Intercloud is that the security allocated to any stream is automatically proportional to the economic value it holds.

\begin{definition}[Stream economic value]
\label{def:stream_value}
The \emph{economic value} of stream $\Stream{i}$, measured in Intercoin, is the dot product of its local coin balances with their exchange rates:
\begin{equation}
  V_i = \sum_{c} \Stream{i}.\mathit{balance}[c] \cdot \mathit{exRate}(c).
\end{equation}
By the backing invariant (Definition~\ref{def:backing}), $\INTER_i \geq V_i$ at all times.
\end{definition}

\begin{theorem}[Security tracks value]
\label{thm:security_tracks_value}
Let $\Stream{i}$ hold economic value $V_i$. Then:
\begin{enumerate}[label=(\roman*)]
  \item $\INTER_i \geq V_i$.
  \item $n_i \geq \lceil c\sqrt{V_i} \rceil$.
  \item Double-spending probability is at most $e^{-s}$ where $s \propto \sqrt{V_i}$.
  \item When value $a \cdot \mathit{exRate}(c)$ moves from $\Stream{i}$ to $\Stream{j}$, exactly that amount of Intercoin weight moves with it: $\INTER_j$ increases and $\INTER_i$ decreases by $\Delta = a \cdot \mathit{exRate}(c)$, preserving $\INTER \geq V$ at both streams.
\end{enumerate}
\end{theorem}

\begin{IEEEproof}
(i) Directly from the backing invariant and Definition~\ref{def:stream_value}.
(ii) From Definition~\ref{def:swarmsize}: $n_i = \lceil c\sqrt{\INTER_i} \rceil \geq \lceil c\sqrt{V_i} \rceil$ since $\INTER_i \geq V_i$.
(iii) From Corollary~\ref{cor:proportional}: $s \propto n_i \propto \sqrt{\INTER_i} \geq \sqrt{V_i}$.
(iv) From Theorem~\ref{thm:zerosum}: $\Delta$ is deducted from $\INTER_i$ and added to $\INTER_j$. By Lemma~\ref{lem:backing}, the backing invariant is preserved. The exchange rate is determined by the state of $\Stream{c}$. By Remark~\ref{rem:rate_timing}, it changes only on emission/retirement events or reserve adjustments, which are themselves serialised by $\Stream{c}$'s Executor. The transfer on $\Stream{i}$ is serialised by $\Stream{i}$'s Executor. $\Delta$ is computed and committed atomically with the transfer message appended to $\mathcal{M}_i$ --- the message records both $a$ and the $\mathit{exRate}(c)$ value used. Any subsequent change to $\mathit{exRate}(c)$ does not retroactively alter $\Delta$. The backing invariant holds at each stream at the logical time of the append to that stream.
\end{IEEEproof}

\begin{corollary}[Self-calibrating security]
No operator needs to manually configure how much security a stream receives. As value flows into a stream, Intercoin weight flows with it, the backing invariant ensures $\INTER_i \geq V_i$, and the next epoch shuffle assigns a larger swarm proportional to $\sqrt{\INTER_i}$. Security is a local property of each stream, automatically calibrated to the value it holds.
\end{corollary}

\subsection{Compartmentalisation}

\begin{proposition}[Security compartmentalisation]
\label{prop:compartment}
A successful attack on stream $\Stream{i}$ requires bribing more than $\supermaj n_i$ swarm members. The minimum cost of such a bribe exceeds $\supermaj V_i$.
\end{proposition}

\begin{IEEEproof}
Each corrupt member's stake is at least $\INTER_i / n_i$ (Definition~\ref{def:staking_req}). By Theorem~\ref{thm:rational_incorrupt}, a rational member accepts a bribe only if it exceeds their staked balance. The minimum bribe per member is $> \INTER_i / n_i$. Corrupting a supermajority requires bribing more than $\supermaj n_i$ members:
\begin{equation}
  \mathrm{total\ bribe} > \supermaj n_i \cdot \frac{\INTER_i}{n_i}
    = \supermaj \INTER_i \geq \supermaj V_i,
\end{equation}
where the last inequality uses $\INTER_i \geq V_i$.
\end{IEEEproof}

An attacker must spend more than two-thirds of a stream's value to corrupt it. This is the formal expression of the armoured-convoy principle: the cost of attack scales with the value of the target, not with the size of the global network.

\section{Privacy, Transparency, and the\\Coin--Content Separation}
\label{sec:privacy_separation}

\subsection{Two Orthogonal Layers}

Every stream in Intercloud has two logically distinct layers. The \textbf{coin layer} comprises balances, transfers, exchange rates, emission, and retirement events. These are secured by Watchers who see only hashes. The aggregate flow of coins between streams is observable to anyone watching the integrity layer, enabling regulatory oversight of value movement without revealing individual transaction contents. The \textbf{content layer} comprises the actual messages, documents, agreements, chat histories, ownership records, or disbursement rules encoded in $\mathcal{M}_i$. These are end-to-end encrypted and visible only to parties holding the appropriate decryption keys.

\begin{definition}[Coin--content separation]
\label{def:coin_content}
The \emph{coin layer} of stream $\Stream{i}$ consists of the subset of messages in $\mathcal{M}_i$ that modify balances, trigger emission/retirement events, or adjust the Intercoin reserve. These must produce integrity-layer effects (changes to $h_i$, $\INTER_i$, and exchange rates) visible to Watchers. The \emph{content layer} consists of all remaining messages --- private communications, signed agreements, governance votes, or any other application data, which may be fully encrypted.
\end{definition}

\begin{remark}[Pre-signed coin movements]
A stream's content layer can contain pre-signed authorisations for future coin movements. For example, a stream holding an escrow agreement can contain a cryptographically signed disbursement instruction that, when revealed, triggers a transfer on the coin layer. The \emph{existence} of an agreement is private, but its \emph{execution} on the coin layer is observable.
\end{remark}

\subsection{Regulatory Observability}

\begin{proposition}[Regulatory coin-weight observability]
\label{prop:regulatory}
A regulator with access to the integrity layer can observe, for every transfer between streams $\Stream{i}$ and $\Stream{j}$, the Intercoin weight transferred ($\Delta = a \cdot \mathit{exRate}(c)$), the direction, and the timestamp --- without learning the plaintext amount $a$, the coin type $c$, the identities of stream owners, or any message content.
\end{proposition}

\begin{IEEEproof}
Every valid transfer updates $h_i$ and $h_j$ and moves Intercoin weight $\Delta$ from $\INTER_i$ to $\INTER_j$ (Theorem~\ref{thm:zerosum}). The integrity layer stores $(h_i, \INTER_i, \hash{\mathit{Rules}_i})$, so the change $\Delta\INTER_i = -\Delta$ and $\Delta\INTER_j = +\Delta$ are observable at the integrity layer, as is the timestamp. The pre-established relation $R_{ij}$ is also integrity-layer metadata. The plaintext amount $a$, coin type $c$, stream owner identities, and message content are in the data layer and are not accessible to Watcher nodes.
\end{IEEEproof}

\begin{remark}[Flow graph vs.\ transaction graph]
A regulator sees a directed weighted graph of Intercoin weight flows between stream identifiers, with timestamps. This is weaker than a full transaction graph but stronger than nothing: the topology of economic relationships and the magnitude of value flows are visible. This matches the regulatory need to detect systemic risk and large-scale money flows without individual transaction surveillance.
\end{remark}

This is a strictly weaker transparency requirement than any public blockchain: instead of revealing all transaction data globally, Intercloud reveals only the aggregate flow structure. Communities can publish their own coins and participate in the global economy while their internal transactions remain private.

\subsection{The Global Intercoin Reserve}

Intercoin is the single global reserve currency of the Intercloud network. Its role is to represent \emph{security weight}. When a community issues local coins backed by Intercoin, the Intercoin reserve is locked in the currency stream until coins are retired. The exchange rate formula (Definition~\ref{def:exchange_rate}) ensures that the security weight of a unit of local coin is always exactly proportional to the Intercoin reserve per circulating unit.

\begin{corollary}[Conservation of security weight]
\label{cor:conservation}
The total Intercoin weight across all streams is conserved: $\sum_i \INTER_i = \INTER_{\mathit{total}}$ (a fixed global supply). Security is not created or destroyed --- it moves between streams as coins move.
\end{corollary}

\begin{IEEEproof}
\emph{Base case.} When a new stream is created with initial Intercoin weight $\INTER_i^{(0)}$, that weight is transferred from an existing funding stream. By Theorem~\ref{thm:zerosum}, this transfer preserves the global sum.

\emph{Inductive step.} Assume $\sum_i \INTER_i = \INTER_{\mathit{total}}$ after $k$ transfers. A valid $(k{+}1)$-th transfer preserves $\INTER_i + \INTER_j$ (Theorem~\ref{thm:zerosum}) and leaves all other $\INTER_\ell$ unchanged. By induction, the sum is conserved for all sequences of valid transfers.
\end{IEEEproof}

\section{Comparison and Related Work}
\label{sec:comparison}

Table~\ref{tab:comparison} summarises the key differences between Intercloud and other major decentralised consensus approaches. Intercloud is the only system among these in which per-transaction validation cost is independent of network size and Watcher nodes do not see plaintext.

\begin{table*}[t]
\centering
\small
\caption{Comparison of decentralised consensus systems.}
\label{tab:comparison}
\begin{tabular}{lccccc}
\toprule
System & Global consensus & Per-tx validation & Ledger visible & Double-spend & Max ambiguity \\
\midrule
Bitcoin       & PoW, global  & $O(N)$        & Yes        & Probabilistic    & Unbounded \\
Ethereum      & PoS, global  & $O(N)$ + gas  & Yes        & Probabilistic    & $\leq$ epoch \\
PBFT          & BFT voting   & $O(n^2)$      & Config.    & Deterministic    & None$^\dagger$ \\
Zcash         & PoW, global  & $O(N)$        & Pool only  & Probabilistic    & Unbounded \\
\textbf{Intercloud} & VRF seed & $O(1)$ amort. & Hashes only$^\ddagger$ & $\geq 1-e^{-s}$ & $\leq T_{\epoch}$ \\
\bottomrule
\end{tabular}
\vspace{2pt}
\begin{flushleft}
\footnotesize
$^\dagger$PBFT guarantees termination only under partial synchrony~\cite{castro1999pbft}; under full asynchrony it may not terminate.\\
$^\ddagger$Watchers hold only hashes; coin-weight flows between streams are observable at the integrity layer (\S\ref{sec:privacy_separation}).
\end{flushleft}
\end{table*}

\section{Limitations and Future Work}
\label{sec:limitations}

\subsubsection{Byzantine Executor}
Theorems~\ref{thm:doublespend} and~\ref{thm:capitalflight} assume a single correct (non-Byzantine) Executor per stream. A Byzantine Executor could produce conflicting hashes, generating $\PoC$s against itself; these would accumulate until the swarm is deemed corrupt and the stream goes Red. However, a Byzantine Executor cannot forge a valid transfer without the stream owner's private key, so funds cannot be stolen outright --- only service disrupted. Full Byzantine Executor resistance requires a threshold-signature Executor cluster or a small BFT sub-committee.

\subsubsection{Executor availability}
A crashed Executor makes the stream temporarily unavailable. The append-only log ensures crash recovery to the last consistent state. High-availability Executor deployments (active-passive replication) are straightforward but not formalised here.

\subsubsection{Exchange rate oracle}
Definition~\ref{def:backing} uses a publicly known exchange rate. In practice this rate changes over time and must come from an oracle. If the oracle is compromised, the backing invariant could be violated. A multi-oracle range-agreement mechanism mitigates this risk but is not formalised in the current model.

\subsubsection{Epoch duration}
The bound $T_{\epoch} > D \cdot \delta_{\min}$ depends on the network diameter $D$ and minimum hop time $\delta_{\min}$. Choosing $T_{\epoch}$ too short risks ripple ID collisions across epochs; choosing it too long increases the maximum Yellow duration. Adaptive epoch-duration protocols are future work.

\subsubsection{Network-level privacy}
Even with IP stripping after the first hop, the first-hop IP is visible to the first relay. Full network-level anonymity requires a Mixnet or onion-routing layer, deferred to future work.

\section{Conclusion}
\label{sec:conclusion}

Intercloud achieves decentralised economic consensus with properties that no prior system simultaneously exhibits: sub-linear per-transaction validation cost, hash-only privacy for all Watcher nodes, chilling-effect deterrence replacing Byzantine voting, explicit user-controlled finality via the three-colour state model, and a self-calibrating security economy in which protection is automatically proportional to the value being protected.

The main theorems build on the Magarshak Machine~\cite{magarshak2026mm} and Hoepman's clerk-set bounds~\cite{hoepman2008}: Ripple Termination, Swarm Security, Disagreement Characterisation, Buridan Inapplicability, Rational Incorruptibility, Security Tracks Value, and Conservation of Security Weight form a coherent stack from the network layer to the economic layer. The junior-node corruption detection mechanism and PoC lottery reward system ensure that exposing dishonest nodes is the dominant strategy for every rational participant. Intercoin vesting ensures that a node's stake always exceeds the value it could extract from any single stream it watches.

Local coins, backed by a pluggable exchange-rate formula tied to the ratio of Intercoin reserve to circulating supply, flow across pre-subscribed economic rails. The dot-product of balances and exchange rates determines each stream's security weight automatically. The coin layer and content layer are strictly separated: regulators can follow the flow of value between streams without observing identities, messages, or rules.

The blockchain model deploys an armoured convoy for every dollar. Intercloud deploys a convoy proportional to the square root of the value transported --- and proves that the convoy self-assembles, self-calibrates, and self-polices without any global coordinator. The single global agreement required is an epoch random seed, costing $O(N)$ messages per epoch and amortising to zero per transaction at scale. The Turing Machine defines what can be computed. The Magarshak Machine defines how distributed state evolves safely. Intercloud defines how economies of Magarshak Machines coordinate without a global ledger.

\balance
\bibliographystyle{IEEEtran}
\bibliography{references}

\end{document}